\documentclass{amsart}

\usepackage{graphicx}
\usepackage{latexsym,color}
\usepackage{amssymb}
\usepackage{amsfonts}
\usepackage{amsmath}
\usepackage{algorithm} 
\usepackage[noend]{algpseudocode} 
\usepackage{hyperref}
\usepackage{comment}

\usepackage{float}
\usepackage{pict2e}

\usepackage{mathrsfs}
\usepackage[numbers]{natbib}

\numberwithin{equation}{section}
\newtheorem{thm}{Theorem}[section]

\newtheorem{rem}[thm]{Remark}
\newtheorem{cor}[thm]{Corollary}

\newtheorem{exm}[thm]{Example}

\newcommand{\be}{\begin{equation}}
\newcommand{\ee}{\end{equation}}

\newcommand{\delay}{{D}} 
\newcommand{\hurst}{{\mathcal{H}}} 
\newcommand{\band}{{\mathcal{S}}}  
\newcommand{\antiband}{{\mathcal{T}}} 
\newcommand{\cov}{{\Sigma}} 
\newcommand{\pre}{{\Lambda}} 

\begin{document}
\title{Exponential Utility Maximization in a Discrete Time Gaussian Framework}
 
\author{Yan Dolinsky}
\address{Department of Statistics, Hebrew University of Jerusalem, Israel}
\email{yan.dolinsky@mail.huji.ac.il}

\author{Or Zuk}
\address{Department of Statistics, Hebrew University of Jerusalem, Israel}
\email{or.zuk@mail.huji.ac.il}
\thanks{Y. Dolinsky is supported in part by the GIF Grant 1489-304.6/2019 and the ISF grant 230/21. \\ \hspace*{0.3cm} O. Zuk is supported in part by the ISF grant
2392/22.}
\date{\today}
\maketitle
\begin{abstract}
The aim of this short note is to present a solution 
to the discrete time exponential 
utility maximization problem in a case where the underlying 
asset has a multivariate normal distribution. 
In addition to the usual setting considered in Mathematical Finance, we also consider an investor who is informed about the 
risky asset's price changes with a delay $\delay$. Our method of solution is based on
 the theory developed in \cite{BF:81} and guessing the optimal portfolio. 
\end{abstract}
\begin{description}
\item[Mathematical Subject Classification (2010)] 91B16, 91G10
\item[Keywords] Utility Maximization, Hedging with Delay, Banded Matrices
\end{description}

\keywords{}
 \maketitle \markboth{}{}
\renewcommand{\theequation}{\arabic{section}.\arabic{equation}}
\pagenumbering{arabic}

\section{Introduction and the Main Result}\label{sec:1}\setcounter{equation}{0}

Taking into account frictions is an important challenge in financial modelling.
In this note, we focus on the friction arising from the fact
that investment decisions may be based only on delayed information,
and the actual present market price is unknown at the time of decision making. 
This corresponds to the case where there is a time
delay in receiving market information (or in applying it), which causes the trader’s
filtration to be delayed with respect to the price filtration.
 
We start by briefly reviewing some of the relevant literature. In \cite{S:94} the author solved the mean-variance hedging problem with partial information for European contingent claims in a setup
where the asset price process is a martingale. Subsequently, the study of the mean-variance hedging problem was extended beyond the martingale setup
(see for instance \cite{F:00,MTT:08}).
The work \cite{NO:09} applied Malliavin calculus to the utility maximization problem,
and in \cite{SZ:19} the authors introduce and study a general framework of optimal stochastic control with delayed information.
Another related work is \cite{BD:2021} 
which deals with scaling limits of utility indifference prices with
vanishing delay and risk aversion inversely proportional to the delay.

To formulate our main result, let $(\Omega,\mathcal F,\mathbb P)$ be a complete 
probability space carrying one risky asset which we denote by $S=(S_k)_{0\leq k\leq n}$ where 
$n\in\mathbb N$ is a fixed finite time horizon. We assume that the investor has a bank account that,
for simplicity, bears no interest. We assume that the initial stock price $S_0$ is a known constant 
and the random vector $(S_1,...,S_n)$ has 
a non-degenerate
multivariate normal distribution. Hence, 
\begin{equation*}
(S_1-S_0,S_2-S_1,...,S_n-S_{n-1})\sim \mathcal N(\mu,\cov)
\end{equation*} 
where $\mu=(\mu_1,...,\mu_n)\in\mathbb R^n$ is a constant (row) vector and 
$\cov\in M_n(\mathbb R)$ (as usual, $M_n(\mathbb R)$ denotes the set of all $n\times n$ real matrices) is a constant positive definite matrix. Let $\pre:=\cov^{-1}$ be the precision matrix. 

We fix a nonnegative integer number $\delay\in\mathbb Z_{+}$ and consider a situation where
there is a delay of $\delay$ trading times. Hence, the investor's flow of information is given by the filtration 
$\mathcal G^{\delay}_k:=\sigma\left\{S_0,...,S_{\left(k-\delay\right)^{+}}\right\}$, $k=0,1,...,n$. 
The case $\delay=0$ means that there is no delay and corresponds 
to the usual setting. 
 
A trading strategy is a random vector $\gamma=(\gamma_1,...,\gamma_n)$ which is predictable with respect to the above filtration. Namely,
for any $k$, $\gamma_k$ is a $\mathcal G^{\delay}_{k-1}$ measurable random variable ($\gamma_k$ is the number of shares at time $k-1$).
Denote by $\mathcal A_{\delay}$ the set of all trading strategies. 
For $\gamma\in\mathcal A_{\delay}$ the corresponding portfolio value
at the maturity date is given by 
$$
V^{\gamma}_n=\sum_{i=1}^{n} \gamma_i (S_i-S_{i-1}).
$$
The investor’s preferences are described by an exponential utility function 
$$u(x)=-\exp(-\alpha x), \ \ x\in\mathbb R,$$
 with absolute risk aversion parameter $\alpha>0$, and her goal is to
\begin{equation}\label{2.2}
\mbox{Maximize} \ \ \mathbb E_{\mathbb P}\left[-\exp\left(-\alpha V^{\gamma}_n\right)\right] \ \ \mbox{over} \ \ \gamma\in\mathcal A_{\delay}
\end{equation}
where $\mathbb E_{\mathbb P}$ denotes the expectation with respect to the 
market probability measure $\mathbb P$. Clearly, 
for any portfolio strategy 
$\gamma$ and a constant $\lambda\in\mathbb R$, 
$V^{\lambda\gamma}_n=\lambda V^{\gamma}_n$. 
Thus, without loss of generality, we take the risk aversion $\alpha=1$.

Next, we introduce some notations. Let $\band_{\delay}\subset M_n(\mathbb R)$ be the set 
of all positive definite matrices $Q\in M_n(\mathbb R)$ that satisfy $Q_{ij}=0$ for $|i-j|>\delay$. Namely, 
$\band_{\delay}$ is the set of all banded positive definite matrices with lower bandwidth and upper bandwidth equal to $\delay$. Let $\antiband_{\delay}\subset M_n(\mathbb R)$ be the set of all symmetric matrices $\Gamma\in M_n(\mathbb R)$ that satisfy $\Gamma_{ij}=0$ for $|i-j|\leq \delay.$ 

We arrive at the main result of the paper.
\begin{thm}\label{thm2.1}
There exists a unique decomposition 
\begin{equation}\label{2.3}
\pre=\hat Q^{-1}+\hat\Gamma
\end{equation}
where $\hat Q\in \band_{\delay}$ and $\hat\Gamma\in \antiband_{\delay}$.
The maximizer $\hat\gamma=(\hat\gamma_1,...,\hat\gamma_n)$ for the optimization problem (\ref{2.2}) is unique and is given by the linear form
\begin{equation}\label{2.4}
\hat\gamma_i=\sum_{j=1}^n \pre_{ij} \mu_j -\sum_{j=1}^{i-1} \hat{\Gamma}_{ij} (S_j-S_{j-1}), \ \ i=1,...,n.
\end{equation}
The corresponding value is given by 
\begin{equation}\label{2.5}
\mathbb E_{\mathbb P}\left[-\exp\left(- V^{\hat\gamma}_n\right)\right]=-\sqrt{\frac{|\hat Q|}{|\Sigma|}}\exp\left(-\frac{1}{2}\mu \pre\mu'\right)
\end{equation}
where $'$ denotes transposition and $|\cdot|$ is the determinant of the matrix $\cdot$.
\end{thm}
The proof of Theorem \ref{thm2.1} is presented
in the next section.

\begin{rem}
There are two trivial cases where the optimal trading strategy is deterministic. The first one is the case of no information $\delay=n-1$. 
The second case is 
where $\cov$ is a diagonal matrix, i.e. the stock increments $S_1-S_0,...,S_n-S_{n-1}$ are independent. 
In both of these cases we have $\hat Q=\Sigma$, $\hat \Gamma=0$ and the value is $\exp\left(-\frac{1}{2}\mu \pre\mu'\right)$. In particular, if 
the stock increments $S_1-S_0,...,S_n-S_{n-1}$ are independent 
the optimal (deterministic) strategy does not depend on $\delay$ and given by
$\hat\gamma_i=\frac{\mu_i}{\Sigma_{ii}}$, $i=1,...,n$. 
\end{rem}
For the case $\delay\in\{0,1\}$
the decomposition (\ref{2.3}) is explicit provided that the inverse matrix $\pre$ is known.
\begin{cor}\label{cor1}
${}$\\
(I) For $\delay=0$ (no delay) the optimal strategy is given by 
$$
\hat\gamma_i=\sum_{j=1}^n \pre_{ij} \mu_j -\sum_{j=1}^{i-1} \pre_{ij} (S_j-S_{j-1}), \ \ i=1,...,n.
$$
The value is 
$$\mathbb E_{\mathbb P}\left[-\exp\left(- V^{\hat\gamma}_n\right)\right]=-\frac{1}{\left(|\Sigma|\prod\limits_{i=1}^n \pre_{ii} \right)^{1/2}}\exp\left(-\frac{1}{2}\mu \pre\mu'\right)$$
(II) For $\delay=1$ the optimal strategy is given by 
$$
\hat\gamma_i=\sum_{j=1}^n \pre_{ij} \mu_j -\sum_{j=1}^{i-2} 
\left(\pre_{ij}-\frac{\prod\limits_{k=j}^{i-1} \pre_{k k+1} }{\prod\limits_{k=j+1}^{i-1} \pre_{k k}}\right)
(S_j-S_{j-1}), \ \ i=1,...,n. 
$$
The value is
\begin{eqnarray*}
&\mathbb E_{\mathbb P}\left[-\exp\left(- V^{\hat\gamma}_n\right)\right]
&=-\left(\frac{\prod\limits_{i=2}^{n-1}\pre_{ii}}{|\Sigma|\prod\limits_{i=1}^{n-1}\left(\pre_{ii}\pre_{i+1i+1}-\pre^2_{ii+1}\right) }\right)^{1/2}\exp\left(-\frac{1}{2}\mu \pre\mu'\right).
\end{eqnarray*}
\end{cor}
\begin{proof}
${}$\\
(I) Observe that for $\delay=0$ we have the following properties: 
The matrix $\hat Q$ is diagonal with $\hat Q^{-1}_{ii}=\pre_{ii}$ for all $i$;
The matrix $\hat\Gamma$ is given by $\hat\Gamma_{ij}=\pre_{ij}$ for $i\neq j$ and $\hat\Gamma_{ij}=0$ otherwise. 
This together with (\ref{2.4})--(\ref{2.5}) completes the proof. \\
(II) The matrix $\hat Q$ satisfies $\hat Q_{ij}=0$ for $|i-j|>1$ (i.e. $\hat Q$ is tridiagonal). Hence, from 
(2.2) in \cite{B:1979}
$$
[\hat Q^{-1}]_{ij}=[\hat Q^{-1}]_{ji}=\frac{\prod\limits_{k=j}^{i-1} [\hat Q^{-1}]_{k k+1} }{\prod\limits_{k=j+1}^{i-1} [\hat Q^{-1}]_{k k} } 
\ \ \mbox{if} \ \ i>j+1
$$
and from (3.1) in \cite{B:1979}
$$
|\hat Q^{-1}|=\frac{\prod\limits_{i=1}^{n-1}\left([\hat Q^{-1}]_{ii}[\hat Q^{-1}]_{i+1i+1}-[\hat Q^{-1}]^2_{ii+1}\right)}{\prod\limits_{i=2}^{n-1}[\hat Q^{-1}]_{ii}}.
$$

Since $\hat\Gamma\in\antiband_1$ then
\begin{eqnarray*}
[\hat Q^{-1}]_{ij}=\pre_{ij} \ \ \mbox{if} \ \ |i-j|\leq 1. 
\end{eqnarray*} 
By combining the above equalities with (\ref{2.4})--(\ref{2.5}) we complete the proof. 
\end{proof}
\begin{rem}
Let us emphasize that we do not assume any Markov structure for the underlying asset, hence even for the ``usual'' case $\delay=0$ Theorem \ref{thm2.1}
gives an efficient way to solve a non-Markovian optimization problem.
\end{rem}

We end this section with the following example. 
\begin{exm}\label{exm.1}
Consider a simple model where the increments have zero mean (i.e. $\mu=0$) and the covariance matrix is
the Kac-Murdock-Szeg\"{o} matrix
$$
\Sigma_{ij}:=\rho^{|i-j|} \ \ \forall i,j
$$
for some constant $\rho\in(0,1)$. Assume that $n\geq 3$.

The matrix $\Sigma$ has a determinant equal to 
$|\Sigma|=\left(1-\rho^2\right)^{n-1}$, 
and a simple tridiagonal inverse that is given by 
(see \cite {D:2003})
\[ \pre_{ij}=
\begin{cases}
 0, &|i-j|>1 \\
 -\frac{\rho}{ 1-\rho^2}, & |i-j|=1 \\
\frac{1+\rho^2}{ 1-\rho^2} , &1<i=j<n \\
 \frac{1}{ 1-\rho^2}, & \mbox{otherwise}.
\end{cases}
\]
${}$\\
\textbf{Case I: $\delay=0$.}
From Corollary \ref{cor1} we obtain that the optimal strategy is
$\hat\gamma_1=0$, and $\hat\gamma_i=\frac{\rho}{1-\rho^2} (S_{i-1}-S_{i-2})$ for $i>1$. 
The value is 
$-\sqrt{\frac{1-\rho^2}{(1+\rho^2)^{n-2}}}
$.
\\
\textbf{Case II: $\delay=1$.}
 Corollary \ref{cor1} gives the optimal strategy 
 $$\hat\gamma_1=\hat\gamma_2=0 \ \ \mbox{and} \ \ \ \hat\gamma_i=\frac{1+\rho^2}{1-\rho^2}\sum_{j=1}^{i-2}\left(-\frac{\rho}{1+\rho^2}\right)^{i-j} (S_j-S_{j-1}) \ \ \mbox{for} \ \ i>2.$$
 The value is $-\sqrt{\frac{\left(1-\rho^2\right)\left(1+\rho^2\right)^{n-2}}{\left(1+\rho^2+\rho^4\right)^{n-3}}}.$
 \end{exm}

\section{Proof of Theorem \ref{thm2.1}}\label{sec:2}
In this section we prove Theorem \ref{thm2.1}.
\begin{proof}
The proof will be done in three steps. \\
${}$\\
 \textbf{Step I:}
 In this step we provide an existence and uniqueness for the decomposition (\ref{2.3}). 
We apply the results from \cite{BF:81} which perfectly fit our purposes. 
First we introduce some notations. Let $C \in M_n(\mathbb R)$ and $I,J \in \{1,..,n\}^k$ be two integer vectors of length $k \leq n$ denoting row and column indices, satisfying $I = (i_1,..,i_k), J=(j_1,..,j_k)$ 
with $1 \leq i_1 < i_2 < .. < i_k \leq n$, $1 \leq j_1 < j_2 < .. < j_k \leq n$. 
We denote by $C_J^I$ the minor specified by these indices, i.e. the determinant of the $k \times k$ sub-matrix determined by taking the elements $C_{ij}$ for $i=i_1,..i_k ; j=j_1,..j_k$. For any two natural numbers $b \geq a$ we denote the vector $(a,a+1,...,b)$ by $[a:b]$.

Since $\pre$ is positive definite it follows by Sylvester’s criterion (see e.g. \cite{G:1991}) that its principal minors
$\pre^{[k+1:k+\delay]}_{[k+1:k+\delay]}>0$ for all $k\leq n-\delay-1$ and 
$\pre^{[k:k+\delay]}_{[k:k+\delay]}>0$ 
for all $k\leq n-\delay$.
Hence, by applying Theorem 5.5 in \cite{BF:81} for $r=s=\delay+1$ we obtain that there is a unique invertible matrix $R$ such that $R^{-1}\in\band_{\delay}$ and 
$\pre-R\in \antiband_{\delay}$. Set $\hat Q:=R^{-1}$, $\hat\Gamma=\pre-R$. 
Since $\cov$ is symmetric then from uniqueness it follows that $R$ is also symmetric and the first step is completed. \\
${}$\\
 \textbf{Step II: }
Denote by $\mathcal Q_{\delay}$ the set of all equivalent probability measures $\mathbb Q\sim\mathbb P$ with finite entropy 
 $\mathbb E_{\mathbb Q}\left[\log\left(\frac{d\mathbb Q}{d\mathbb P}\right)\right]<\infty$ relative to $\mathbb P$ that satisfy 
 \begin{equation}\label{4.2-}
 \mathbb E_{\mathbb Q}[S_t-S_s|\mathcal G^{\delay}_s]=0 \ \ \forall t\geq s.
 \end{equation}
In this step we prove the following verification result:
If a triplet $(\tilde\gamma,\tilde{\mathbb Q},C)\in \mathcal A_{\delay}\times \mathcal Q_{\delay}\times\mathbb R$ satisfies
\begin{equation}\label{4.2-+}
V^{\tilde\gamma}_n+\log\left(\frac{d{\tilde{\mathbb Q}}}{d\mathbb P}\right)=C
\end{equation}
 then $\tilde\gamma\in\mathcal A_{\delay}$ is the unique optimal portfolio for the optimization 
 problem (\ref{2.2}) and the corresponding value is 
 $
\mathbb E_{\mathbb P}\left[-\exp\left(- V^{\tilde\gamma}_n\right)\right]=-e^{-C}.$

Indeed, using similar arguments as in the proof of Lemma 2.1 in \cite{BD:2021}, it follows that for any $\gamma\in\mathcal A_{\delay}$ and 
 $\mathbb Q\in \mathcal Q_{\delay}$ 
 we have 
 $\mathbb E_{\mathbb Q}[V^{\gamma}_n]=0$ provided that 
 $\mathbb E_{\mathbb P}\left[\exp\left(-V^{ \gamma}_n\right)\right]<\infty$.
 From (\ref{4.2-+})
 $\mathbb E_{\mathbb P}\left[\exp\left(-V^{\tilde \gamma}_n\right)\right]=e^{-C}.$ Hence,
 $\mathbb E_{\tilde{\mathbb Q}}[V^{\tilde\gamma}_n]=0$ and 
by applying (\ref{4.2-+}) again 
\begin{equation}\label{4.2}
\log\left(\mathbb E_{\mathbb P}\left[\exp\left(-V^{\tilde \gamma}_n\right)\right] \right)=-C=- \mathbb E_{\tilde{\mathbb Q}}\left[\log
\left(\frac{d\tilde{\mathbb Q}}{d\mathbb P}\right)\right].
\end{equation}

Next, it is well known that for a strictly concave utility maximization problem, the optimizer is unique (if exists). Thus, in view of (\ref{4.2}), in order to complete the proof of the second step it remains to show that in general we have the inequality
\begin{equation}\label{4.2+}
\log\left(\mathbb E_{\mathbb P}\left[\exp\left(-V^{\gamma}_n\right)\right] \right)\geq -\mathbb E_{\mathbb Q}\left[\log\left(\frac{d\mathbb Q}{d\mathbb P}\right)\right] \ \ \forall (\gamma,\mathbb Q)\in\mathcal A_{\delay}\times\mathcal Q_{\delay}.
\end{equation}
Let us establish (\ref{4.2+}).
Without loss of generality assume that $\mathbb E_{\mathbb P}\left[\exp\left(-V^{ \gamma}_n\right)\right]<\infty$ (otherwise (\ref{4.2+}) is trivial). Then
for any $z\in\mathbb R$
\begin{align*}
\mathbb E_{\mathbb P}\left[\exp\left(-V^{ \gamma}_n\right)\right]
&=\mathbb E_{\mathbb P}\left[\exp\left(-V^{ \gamma}_n\right)+z \frac{d{\mathbb Q}}{d\mathbb P}V^{\gamma}_n\right]\\
&\geq \mathbb E_{\mathbb P}\left[z\frac{d{\mathbb Q}}{d\mathbb P}\left(1-\log
\left(z\frac{d{\mathbb Q}}{d\mathbb P}\right)\right)\right]\\
&=z-z\log z-z\mathbb E_{\mathbb Q}\left[\log
\left(\frac{d{\mathbb Q}}{d\mathbb P}\right)\right].
\end{align*}
The first equality is due to $\mathbb E_{{\mathbb Q}}[V^{\gamma}_n]=0$.
The inequality follows from 
the Legendre-Fenchel duality inequality $xy \leq e^x+y(\log y-1)$ for all $x,y \in \mathbb{R}$ by setting $x=-V^{\gamma}_n$ and 
$y=z\frac{d\mathbb Q}{d\mathbb P}$.
The last equality is straightforward. 
From simple calculus it follows that the concave function 
$z\rightarrow z-z\log z-z\mathbb E_{\mathbb Q}\left[\log
\left(\frac{d{\mathbb Q}}{d\mathbb P}\right)\right]$, $z>0$ attains its maximum at $z^{*}:=\exp\left(-\mathbb E_{\mathbb Q}\left[\log
\left(\frac{d{\mathbb Q}}{d\mathbb P}\right)\right]\right)$ and the corresponding maximal value 
is also $z^{*}$. This completes the proof of (\ref{4.2+}).
\\
${}$\\
 \textbf{Step III:}
In view of Step I the portfolio $\hat\gamma$ from (\ref{2.4}) is well defined. 
From the fact that 
$\hat\Gamma_{ij}=0$ for $|i-j|\leq \delay$ we obtain that $\hat\gamma\in\mathcal A_{\delay}$. 
Set 
\begin{equation*}
C:=\frac{1}{2}\left(\log |\Sigma|-\log |\hat Q|+ \mu \pre\mu' \right) 
\end{equation*}
and define the measure $\hat{\mathbb Q}$ by the Radon-Nikodym derivative
\begin{equation}\label{4.3}
\frac{d\hat{\mathbb Q}}{d\mathbb P}:=\exp\left(C-V^{\hat\gamma}_n\right).
\end{equation}
From Step II it follows that in order to complete the proof of Theorem \ref{thm2.1} it remains to establish 
that 
the measure $\hat{\mathbb Q}$ is a probability measure which satisfies $\hat{\mathbb Q}\in\mathcal Q_{\delay}$.

To that end 
define $X=(X_1,...,X_n)$ by $X_i=S_i-S_{i-1}$, $i=1,...,n$. 
From (\ref{2.4}) we have (recall that $\hat\Gamma$ is symmetric)
 $V^{\hat\gamma}_n=\hat\gamma X'=\mu \pre X'-\frac{1}{2} X\hat\Gamma X'$. This together with (\ref{2.3}) and
 (\ref{4.3}) gives 
 $$
\frac{d\hat{\mathbb Q}}{d\mathbb P}=\frac{\frac{\exp\left(-\frac{1}{2}X\hat Q^{-1}X'\right)}{\sqrt {(2\pi)^n |\hat Q|}}}{\frac{\exp\left(-\frac{1}{2}(X-\mu)\pre(X-\mu)'\right)}{\sqrt {(2\pi)^n |\Sigma|}}}.$$
Hence, the relation $\left(X;{\mathbb P}\right) \sim\mathcal N(\mu,\Sigma)$
implies that $\hat{\mathbb Q}$ is a probability measure 
and $\left(X;\hat{\mathbb Q}\right) \sim\mathcal N(0,\hat Q)$.
In particular $\mathbb E_{\hat{\mathbb Q}}\left[ \log\left(\frac{d\hat{\mathbb Q}}{d\mathbb P}\right)\right]<\infty$. 
Finally, from 
the property 
 $\hat Q_{ij}=0$ for $|i-j|>\delay$ we obtain that (under the probability measure $\hat{\mathbb Q}$)
 $X_k$ is independent of $(X_1,...,X_{(k-1-\delay)^{+}})$ for all $k$. This yields
 (\ref{4.2-}) and completes the proof. 
 \end{proof}
 
We end this section with the following remark. 
\begin{rem}
Observe that for any positive definite matrix $Q\in M_n(\mathbb R)$, the probability 
measure $\mathbb Q\sim\mathbb P$ which is given by 
 $$
\frac{d{\mathbb Q}}{d\mathbb P}=\frac{\frac{\exp\left(-\frac{1}{2}X Q^{-1}X'\right)}{\sqrt {(2\pi)^n | Q|}}}{\frac{\exp\left(-\frac{1}{2}(X-\mu)\pre(X-\mu)'\right)}{\sqrt {(2\pi)^n |\Sigma|}}}$$
satisfies 
$\left(X;{\mathbb Q}\right) \sim\mathcal N(0, Q)$. In particular for a diagonal $Q$ the measure $\mathbb Q$ is an equivalent martingale measure, and
so the market is arbitrage-free in the classical sense (no delay).  Clearly, for any $\delay$ the set $\mathcal Q_{\delay}$ contains the set of all equivalent martingale measures. 
\end{rem}
\section{Computational Results}\label{sec:3}\setcounter{equation}{0}
We start this section by providing 
 an explicit algorithm for the computation of $\hat Q$ and $\hat\Gamma$ from (\ref{2.3}). Recall the notations from Sections \ref{sec:1},\ref{sec:2}.
In view of Corollary 3.2 in \cite{BF:81} and the fact that $\hat Q\in \band_{\delay}$, we look for a matrix 
$R:=\hat Q^{-1}$ that satisfies $R_{ij}=\pre_{ij}$ for $i,j<\delay+1$ and
has vanishing super and sub $\delay+1$ minors, i.e. minors of the form 
$R^I_J$ where $I=(i_1,...,i_{\delay+1})$, $J=(j_1,...,j_{\delay+1})$ that satisfy one of the inequalities 
$j_1>i_{\delay+1}-\delay$ (super minors) or $i_1>j_{\delay+1}-\delay$ (sub minors). 
The matrix $R$ is symmetric hence we can restrict ourselves to looking only at super minors. Since $\pre$ is positive definite and $R_{ij}=\pre_{ij}$ for $i,j<\delay+1$, it follows that the principal minors 
$R^{[i+1:i+\delay]}_{[i+1:i+\delay]}>0$ for all $i\leq n-\delay-1$.
Thus, from Theorem 3.4 in \cite{BF:81} we conclude that the matrix $R:=\hat Q^{-1}$ is the unique 
symmetric matrix that satisfies 
$R_{ij}=\pre_{ij}$ for $|i-j|\leq \delay$ and
\begin{equation}\label{3.1}
R^{[i:i+\delay]}_{(i+1,i+2,...,i+\delay,i+m)}=0, \ \ i\in [1:n-\delay-1], \ \ m\in [1+\delay:n-i].
\end{equation}

Observe that it is possible to order the $(\delay+1)$-minors in \eqref{3.1} such that each minor is a function of $(\delay+1)^2-1$ known elements of $R$, and one unknown element, hence providing a linear equation for this unknown element. 
For example, the first minor $R^{[1:1+\delay]}_{[2:2+\delay]}$ is determined by $(R_{ij} ), i=1,..,\delay+1 ; j=2,..,\delay+2$ which are all set to the corresponding elements of $\pre$ because $|i-j|\leq \delay$ except for $R_{1 \delay+2}$. 
For four integers $a \leq i \leq b \leq j$, denote by $[a:b]^{i \to j}$ the vector obtained by omitting the index $i$ from $[a:b]$, and adding in the last coordinate the index $j$, i.e. $[a:b]^{i \to j} = (a,a+1,..,i-1,i+1,..,b, j)$. 
Hence, in the Laplace expansion for the $(\delay+1)$-minor
\begin{equation*}
0 = R^{[1:1+\delay]}_{[2:2+\delay]} = (-1)^{\delay}R_{1 2+\delay}R^{[2:1+\delay]}_{[2:1+\delay]}
+\sum_{j=1}^{\delay} (-1)^{1+j} R_{1 1+j} R^{[2:1+\delay]}_{[2:2+\delay]^{1+j \to 2+\delay}}
\end{equation*}
all the $\delay$-minors are known, and the unknown term $R_{1 2+\delay}$ appears only in the first term. 
It is therefore possible to express $R_{1 2+\delay}$ uniquely as a linear function of the known $(\delay+1)^2-1$ $R_{ij}$ values. 

Once $R_{1 \delay+2}$ is determined, we can express the next $(\delay+1)$-minor $R^{[2:2+\delay]}_{[3:3+\delay]}$ in terms of $(\delay+1)^2-1$ known $R_{ij}$ elements and the unknown element $R_{2 \delay+3}$, and determine it in a similar manner. 
Algorithm \ref{alg:band_decomposition} describes how to proceed and determine all the unknown elements $R_{ij}$ one by one, hence it provides the entire decomposition of $\pre$ into $R + \hat\Gamma$ (recall that $R=\hat Q^{-1}$).
\begin{algorithm} 
	\caption{Banded Matrix Decomposition}
	\begin{algorithmic}[1]
		\Statex {\bf Input:} $\Sigma \in M_n(\mathbb R)$ - a real Positive Definite matrix, $\delay \leq n-1$ - band width.
		\Statex {\bf Output:} Two matrices $R+\hat\Gamma=\cov^{-1}$ such that $R^{-1} \in \band_{\delay}, \hat\Gamma \in \antiband_{\delay}$.
		\State Compute the inverse $\pre=\Sigma^{-1}$ and set $R_{ij} = \pre_{ij}, \forall |i-j|\leq \delay$.
		\For{$m=\delay+1$ to $n-1$} 
		\For{$i=1$ to $n-m$} 
		\State Set $$R_{i+m i}, R_{i i+m} = \frac{\sum\limits_{j=1}^{\delay} (-1)^{j+\delay} R_{i i+j} R_{[i+1:i+\delay]^{i+j \to i+m}}^{[i+1:i+\delay]}}{R_{[i+1:i+\delay]}^{[i+1:i+\delay]}}.$$
		\EndFor	 
		\EndFor	 
		\State Set $\hat\Gamma = \pre-R$.
		\State Output the matrices $R, \hat\Gamma$. 
\end{algorithmic}
\label{alg:band_decomposition}
\end{algorithm}

\begin{rem}
As matrix inversion can be performed in $O\left(n^{2.373}\right)$ floating points operations \cite{W:2011}, the total number of floating points operations for Algorithm \ref{alg:band_decomposition} is $O\left(n^{2.373}\right)$ (for matrix inversion), plus $O\left(n^2 \times \delay^{2.373}\right)$ for computing the minors (see \cite{A:1974}), hence yielding an overall computational complexity of $O\left(n^2 (n^{0.373} + \delay^{2.373})\right)$ floating point operations. 
When $\Sigma$ has a special structure, specialized faster algorithms can be used. For example, if $\Sigma$ is a Toeplitz matrix, inversion can be performed in $O\left(n^2\right)$ floating point operations using Trench's algorithm \cite{T:1964}, hence the overall complexity is reduced to $O\left(n^2 \delay^{2.373}\right)$. 
\end{rem}

Next, we implemented Algorithm \ref{alg:band_decomposition} in $R$ code, freely available at \href{https://github.com/orzuk/BandedDecomposition}{github}.
Using this implementation, we present some numerical results for the simple model in Example \ref{exm.1}, and for a more realistic model of a discretized fractional Brownian motion. 
A fractional Brownian motion (fBm) \cite{M:1968} with Hurst parameter $\hurst \in (0,1)$ is a centered Gaussian process $B^{\hurst}_t$, $t\geq 0$ with the covariance function 
\be
Cov\left(B^\hurst_s,B^\hurst_t\right)=\frac{1}{2}
 \left(t^{2\hurst}+s^{2\hurst}-|t-s|^{2\hurst}\right).
\label{eq:fbm} 
\ee
If $\hurst=\frac{1}{2}$ we recover the standard
Brownian motion. Models of asset prices based on fBm have long attracted the interest of
researchers for their properties of long-range dependence, see for instance, \cite{L:91,W:99}. 
For some explicit computations for portfolio optimization in the fBm framework see \cite{G:19,G:21}.

We consider a discretized version of the fBm framework. Namely, we fix $\hurst\in (0,1)$, a time-resolution $\Delta t > 0$ and a final time $T=n \Delta t$ for some  $n\in\mathbb N$, and consider the fBm market which is active at times 
$0, \Delta t, 2\Delta t, \dots, n\Delta t=T$. The stock price at time $k \Delta t$ is given by 
$S_k=S_0+B^{\hurst}_{k \Delta t}$, $k=0,1,\dots,n$. In this case 
$\mu=0$ and by \eqref{eq:fbm} the covariance matrix is the Toeplitz matrix given by 

\begin{align}
\Sigma_{ij} &= Cov\left(B^{\hurst}_{(i+1)\Delta t}-B^{\hurst}_{i \Delta t}, B^{\hurst}_{(j+1)\Delta t}-B^{\hurst}_{j\Delta t}\right) \nonumber \\
&=\frac{\Delta t^{2\hurst}}{2} \left(|i-j-1|^{2\hurst}+|i-j+1|^{2\hurst}-2|i-j|^{2\hurst}\right). 
\label{eq:fbm_discrete}
\end{align}
By following carefully the proof of Proposition 1.6 in \cite{N:2012} we obtain that $\Sigma$ is positive definite.
Moreover, when $\Delta t$ is changed $\Sigma$ is multiplied by a constant scalar, hence for fixed values of $n$, the delay $\delay$ and the $\hurst$
parameter, the value is invariant to $\Delta t$.

We computed numerically the value as a function of the delay $\delay$ for both models, for $n=64$ and $\mu=0$. The results, displayed in Fig. \ref{fig:value_vs_H}, show the decrease in the value as $\delay$ in increasing for both models. We first examined the Kac-Murdock-Szeg\"{o} covariance matrix $\Sigma_{ij}:=\rho^{|i-j|}$ from Example \ref{exm.1}. As expected, for any delay $\delay$ the value is monotonically increasing with the correlations parameter $\rho$, hence longer delays can be tolerated with only a negligible drop in the value as $\rho$ increases. For example, when there is no delay the value is significantly reduced compared to the maximum $0$ only for $\rho<0.3$, while for a delay of $\delay=4$ a significant drop in the value is observed as soon as $\rho<0.9$. 

Next, we examined the covariance matrix representing the discretized fractional Brownian motion from \eqref{eq:fbm_discrete}.
For high $\hurst$ values ($\hurst>0.5$), the increments are positively correlated hence the value increases with $\delay$ towards the maximum $0$.
For $\hurst=0.5$ we reduce to the independent Brownian motion case, hence the value drops to $-1$ regardless of the delay.
For low $\hurst$ values ($\hurst<0.5$), the increments are negatively correlated hence the value increases, but the increase is non-monotonic in $\hurst$. The reason is that for very low $\hurst$ values the correlation between the increments decays faster to $0$ with their time distance, hence a delay results in almost complete loss of information regarding the current price.

The entire computations for each model at a resolution of $0.01$ of $\rho$ or $\hurst$ (in total computing the value for $100$ parameter values $\times$ $8$ values of $\delay$, i.e. $800$ different $\Sigma$ matrices) took $\sim\!200$ seconds on a standard PC laptop in our implementation.
While the value is invariant to the total time $T$, changing the discretization parameter $n$ may alter the value. Increasing $n$ can be interpreted as either keeping a fixed time horizon $T$ and decreasing $\Delta t$, or as fixing $\Delta t$ and increasing the horizon $T$.
For a fixed delay $\delay$, the value increases towards zero with $n$, this make sense since the continuous time fBm model 
leads to arbitrage (see for instance \cite{R:1997, C:2003}).
  For example, as shown in Figure \ref{fig:value_vs_H}, for $\hurst=0.2$ and $D=1$ the value is roughly $-0.426$. The value increases to $-0.123$ and $-0.007$ for $n=128$ and $256$, respectively. Increasing the delay $\delay_n$ with $n$ while keeping the ratio $\frac{\delay_n}{n}$ fixed amounts to refining the discretization towards the continuous limit with a fixed continuous delay $0 < \delay < T$. The value appears to converge to a negative number for this case (i.e. for $\frac{\delay_n}{n}=\frac{1}{64}$), with $-0.283$ for $n=256$ and $-0.228$ for $n=1024$.
A more thorough theoretical and numeric investigation of the scaling of the value as a function of the problem's parameters for this and other models remains for future work. 

\begin{figure}[!ht]	\includegraphics[width=0.7\columnwidth]{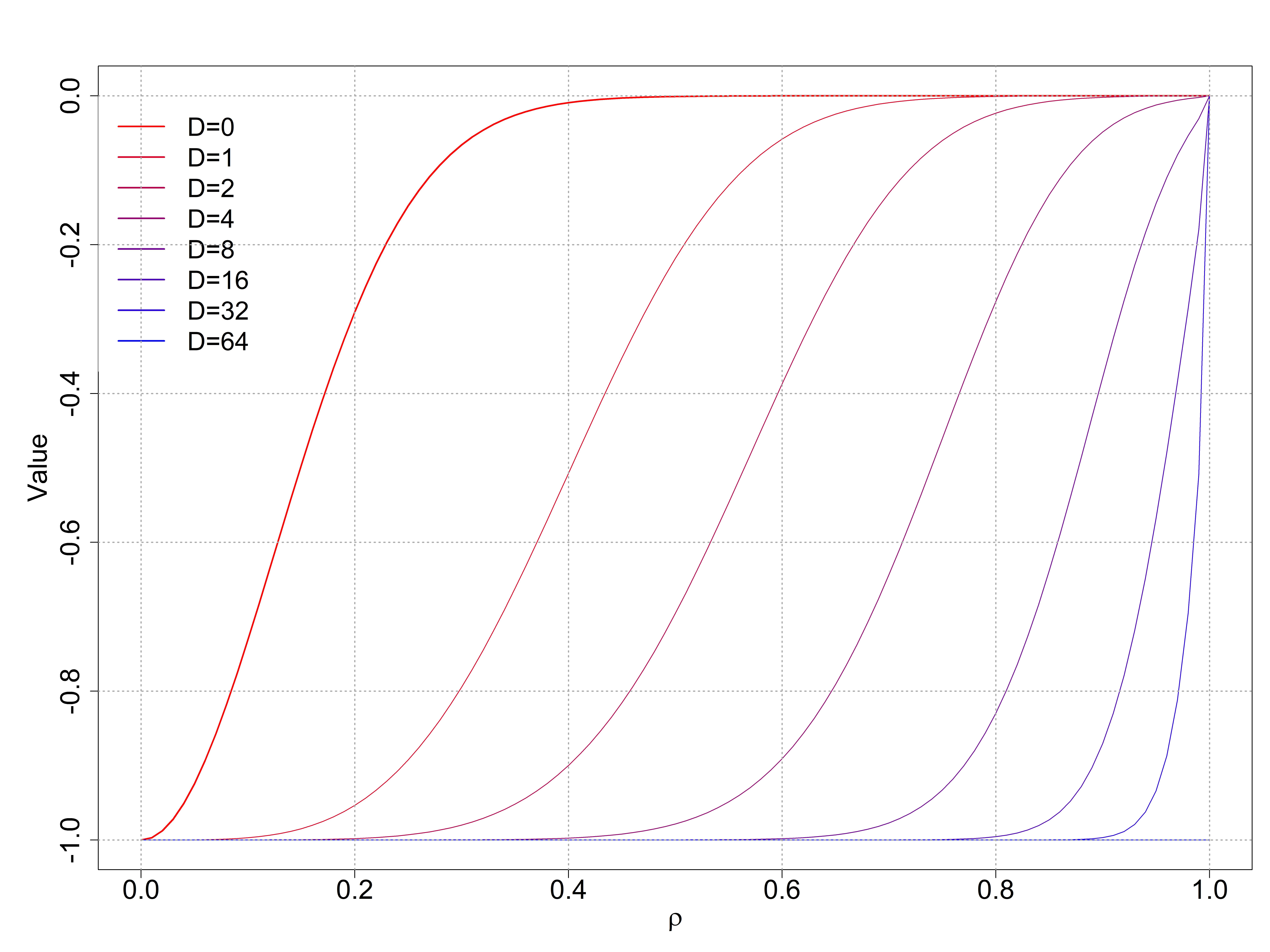} 		\includegraphics[width=0.7\columnwidth]{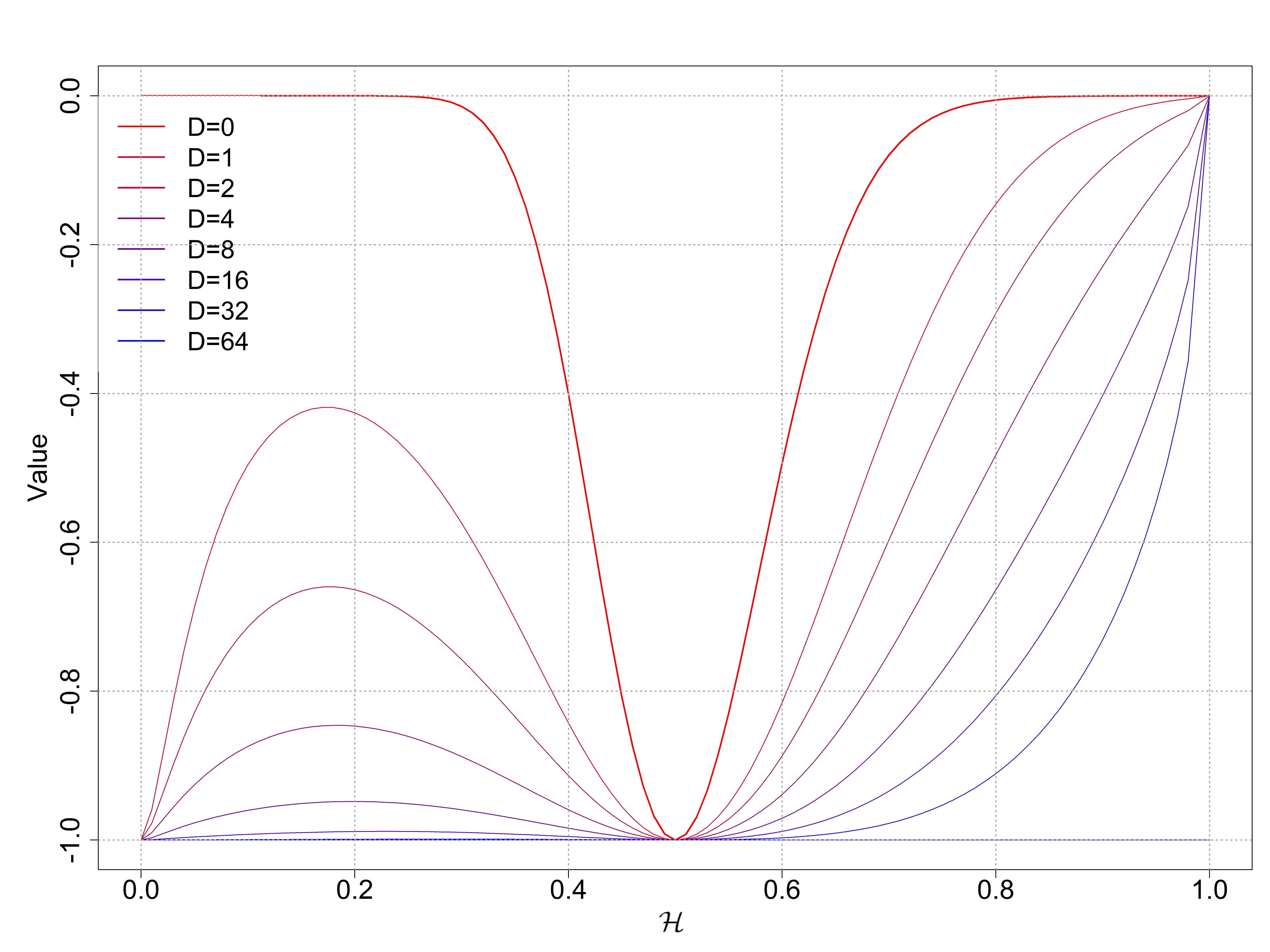} 
	\caption{\footnotesize  The value for different delays $\delay$ (shown in different colors) for $n=64$, $\mu=0$.  
 {\bf Top:} The value is shown as a function of the $\rho$ parameter (x-axis) for the Kac-Murdock-Szeg\"{o} covariance matrix $\Sigma_{ij}:=\rho^{|i-j|}$ from Example \ref{exm.1}. The results for $\delay=0,1$ can be deduced from the closed-form derivations shown in the example. 
{\bf Bottom:} The value is shown as a function of the Hurst index $\hurst$ (x-axis) for a covariance matrix representing the discretized fractional Brownian motion from \eqref{eq:fbm_discrete}.} 
\label{fig:value_vs_H}		
\end{figure}

\section*{Acknowledgments}
We thank the anonymous referees whose comments improved the quality of the paper.

\clearpage
\newpage

\bibliographystyle{plainnat} 
\bibliography{finance}

\end{document}